\documentclass[reqno, 11pt]{amsart}

\usepackage[top=4cm, bottom=3cm, left=3.2cm, right=3.2cm]{geometry}

\newtheorem{thm}{Theorem}

\newtheorem{lem}{Lemma}
\newtheorem{prop}{Proposition}
\newtheorem{exam}{Example}
\theoremstyle{definition}

\theoremstyle{remark}
\newtheorem{rem}{Remark}

\newcommand{\dint}{\displaystyle\int}

\newcommand{\bx}{x}

\numberwithin{equation}{section} \numberwithin{lem}{section}
\numberwithin{thm}{section} \numberwithin{prop}{section}
\numberwithin{cor}{section} \numberwithin{rem}{section}

\begin{document}

\title[Exact criterion for Degenerate Keller-Segel system]{Exact criterion for global existence and blow up to a degenerate Keller-Segel system}

\author{Li Chen}
\address{Department of Mathematical Sciences, Tsinghua University,
Beijing, 100084, People's Republic of China}
\email{lchen@math.tsinghua.edu.cn}
\thanks{Li Chen is partially supported by National Natural Science
Foundation of China (NSFC), grant number 11271218, 11011130029}

\author{Jinhuan Wang}\thanks{Corresponding author: Jinhuan Wang,~School of Mathematics, Liaoning University, Shenyang, 110036}
\address{School of Mathematics, Liaoning University, Shenyang 110036\quad
Tel: +86 024 62202209\quad Fax: +86 024 62202209, P. R. China and
Department of Mathematical Sciences, Tsinghua University,
Beijing, 100084, P. R. China}
\email{jhwang@math.tsinghua.edu.cn}
\thanks{Jinhuan Wang is partially supported by National Natural Science
Foundation of China (Grant no. 11126128, 11301243), China Postdoctoral Science Foundation (Grant no. 20110490409)}

\maketitle
\date{}

\begin{abstract}
A degenerate Keller-Segel system with diffusion exponent $m$ with $\frac{2n}{n+2}<m<2-\frac{2}{n}$ in multi dimension is studied. An exact criterion for global existence and blow up of solution is obtained. The estimates on $L^{\frac{2n}{n+2}}$ norm of the solution play important roles in our analysis. These estimates are closely related to the optimal constant in Haddy- Littlewood- Sobolev inequality. In the case of initial free energy less than a universal constant which depends on the inverse of total mass, there exists a constant such that if the $L^{\frac{2n}{n+2}}$ norm of initial data is less than this constant, then the weak solution exists globally; if the $L^{\frac{2n}{n+2}}$ norm of initial data is larger than the same constant, then the solution must blow-up in finite time. Our result shows that the total mass, which plays the deterministic role in two dimension case, might not be an appropriate criterion for existence and blow up discussion in multi-dimension, while the $L^{\frac{2n}{n+2}}$ norm of the initial data and the relation between initial free energy and initial mass are more important.
\end{abstract}


\section{Introduction}
In this article, we will study a degenerate Keller-Segel system for $n\geq 3$
dimension:
\begin{eqnarray}  \label{equs}
\left\{
\begin{array}{llll}
\smallskip
&\rho_t=\Delta \rho^m -{\rm div}(\rho\nabla c),&& x\in \mathbb R^n,\,t\geq 0,\\
&-\Delta c= \rho,&&  x\in \mathbb R^n,~t\geq 0,\\
&\rho(x,0)=\rho_{0}(x), && x\in \mathbb R^n,
\end{array}
\right.
\end{eqnarray}
where diffusion exponent $m\in (\frac{2n}{n+2},2-\frac{2}{n})$, $\rho(x,t)$ represents the density of bacteria and $c(x,t)$ represents the chemical substance concentration. Mass conservation of the system implies $\|\rho(\cdot,t)\|_{L^1}=\|\rho_0(\cdot)\|_{L^1}=M_0$.

Keller-Segel system with linear diffusion was proposed by Patlak \cite{Patlak} and Keller-Segel \cite{KS1,KS2}.
 It is used to describe the collective motion of cells or the evolution of the density of bacteria. This model plays important roles in the study of chemotaxis in mathematical biology. Since 1980, Keller-Segel system was widely studied in the literature. From the work by Childress \cite{Childress}, we known that the behavior of this model strongly depends on the space dimension, the readers are referred to two surveys given by Horstmann \cite{H1,H2}.

Recently, many mathematicians are interested in finding the criterion for global existence and blow up of solution to Keller-Segel type systems. In particular, the 2-dimensional case has been well studied. It is well known that $8\pi$ is the critical mass of 2-dimensional Keller-Segel system \cite{DP,Perthamebook,BDP}. More precisely, if the initial mass $M_0<8\pi$, then there exists global weak solution; if $M_0>8\pi$, then the solution blows up in finite time; The more delicate case $M_0=8\pi$ was studied in \cite{BCC10,BCM08}.

In dimension $n\geq 3$, one has to use nonlinear diffusion to balance the nonlocal aggregation effect. A natural question is to find a criterion for initial data to separate the global existence and finite time blow up to degenerate Keller-Segel system (\ref{equs}) with diffusion exponent $m>1$.

There were two critical diffusion exponents of (\ref{equs}) which have been studied recently. One is that $m^*=2-\frac{2}{n}$, which came from the scaling invariance of the total mass. The following results were obtained in \cite{S1,SK}. If $m>m^*$, the solution exists globally for any initial data; if $1<m\leq m^*$, both global existence and blow-up can happen for some initial data. Later on, Blanchet-Carrillo-Laurencot in \cite{BCL09} studied the degenerate system with diffusion exponent $m=m^*$, a critical mass was given there. Another critical exponent of (\ref{equs}), $m_c=\frac{2n}{n+2}$ was given in \cite{CLW1}, which came from the conformal invariance of the free energy. The authors in \cite{CLW1} showed that $L^{m_c}$ norm of a family of positive stationary solution can be viewed as the criterion for the global existence and blow up of solutions.

In this paper we are interested in finding a criterion to classify the initial data to get either global existence or blow up of the solution. Our analysis will work for all the diffusion exponents $m$ such that $\frac{2n}{n+2}=m_c<m<m^*=2-\frac{2}{n}$.

There are two very important quantities of system \eqref{equs}. One is the total mass which is time independent,
\begin{eqnarray*}
\dint_{\mathbb R^n} \rho(x,t)dx = \dint_{\mathbb R^n} \rho_0(x) dx =M_0,
\end{eqnarray*}
the other is the free energy
\begin{eqnarray*}
{\mathcal F}(\rho)=\dfrac{1}{m-1}\int_{\mathbb R^n}\rho^m( x,t)dx - \dfrac{1}{2}\int_{\mathbb R^n} \rho(x,t)c(x,t)d x,
\end{eqnarray*}
which decays in time due to the following entropy-entropy production relation
\begin{eqnarray*}
\dfrac{d}{dt} \mathcal{F} (\rho(\cdot, t)) +\dint_{\mathbb{R}^n} \rho \Big|\nabla(\dfrac{m}{m-1}\rho^{m-1}-c)\Big|^2 dx =0.
\end{eqnarray*}

The main result of this paper is
\begin{thm}\label{thm}
Assume that the initial density $\rho_0\in L^1_+(\mathbb R^{n})\cap L^{m} (\mathbb R^{n})$ and ${\mathcal F}(\rho_0)<{\mathcal F}^*$, the following holds,
\begin{enumerate}
\item
If $\|\rho_0\|_{L^{\frac{2n}{n+2}} (\mathbb R^{n})}< (s^*)^{\frac{n-2}{2n(m-1)}}$, then \eqref{equs} has a global weak solution, i.e. for all $T>0$, there is a function $\rho(x,t)$ with (for some $1<r,s\leq 2$),
\begin{align*}
&\rho\in L^\infty(0,+\infty; L^1_+(\mathbb R^{n})\cap L^{m}(\mathbb R^n)),\\
& \nabla \rho \in L^2(0,T;L^r(\mathbb R^n)), \quad \partial_t\rho\in L^2(0,T;W^{-1,s}_{\rm loc}(\mathbb R^n)),
\end{align*}
such that it satisfies \eqref{equs} in the sense of distribution.
\item
If $\|\rho_0\|_{L^{\frac{2n}{n+2}}(\mathbb{R}^n)}>(s^*)^{\frac{n-2}{2n(m-1)}}$ and $\rho_0$ has finite second moment, $\rho(x,t)$ is a solution of \eqref{equs}, then there exists a $T^*>0$ such that
\begin{eqnarray}\label{Lmblowup}
\displaystyle\lim_{t\rightarrow T^*}\|\rho(\cdot,t)\|_{L^m(\mathbb{R}^n)}=+\infty.
\end{eqnarray}
\end{enumerate}
Here ${\mathcal F}^*$ and $s^*$ are universal constants given by
\begin{eqnarray}
\label{F*}{\mathcal F}^*&=&\dfrac{2-\frac{2}{n}-m}{(m-1)(1-\frac{2}{n})}\Big(\frac{2n^2\alpha(n)}{C(n)}\Big)^{\frac{n(m-1)}{2n-2-mn}}M_0^{\frac{2n-m(n+2)}{2n-2-mn}}>0,\\ \label{s*} s^*&=&\Big(\dfrac{2n^2\alpha(n)M_0^{\frac{2n-m(n+2)}{n-2}}}{C(n)}\Big)^{\frac{n(m-1)}{2n-2-mn}}>0,
\end{eqnarray}
where $M_0$ is the initial mass $\|\rho_0\|_{L^1(\mathbb{R}^n)}$, $\alpha(n)=\frac{\pi^{n/2}}{\Gamma(\frac{n}{2}+1)}$ is the volume of the unit ball of $\mathbb{R}^n$ and $C(n)$ is the best constant of Hardy-Littlewood-Sobolev inequality, see (\ref{Cn}).
\end{thm}

\begin{rem}
The result does not hold for $m=m^*=2-\frac{2}{n}$, thus there is no contradiction with the result by Blanchet \emph{et al}. in \cite{BCL09}, where a critical mass was obtained.
\end{rem}

\begin{rem}
    The conditions $\rho_0\in L^1_+(\mathbb R^{n})\cap L^{m} (\mathbb R^{n})$ and $\|\rho_0\|_{L^{\frac{2n}{n+2}} (\mathbb R^{n})}< (s^*)^{\frac{n-2}{2n(m-1)}}$ for the existence result imply that the initial free energy is positive, i.e. ${\mathcal F}(\rho_0)> 0$, which can be easily checked by direct computations. Conversely, if the initial free energy is negative, i.e. ${\mathcal F}(\rho_0)< 0$ and $\rho_0\in L^1_+(\mathbb R^{n})\cap L^{m} (\mathbb R^{n})$, then $\|\rho_0\|_{L^{\frac{2n}{n+2}} (\mathbb R^{n})}> (s^*)^{\frac{n-2}{2n(m-1)}}$.  Therefore, our result on the blow-up of solutions allows more initial data than those in the work by Sugiyama. Thus the blow up result improves her work with $\gamma=0$. (In \cite{S1}, Y. Sugiyama proved that if the initial free energy is negative and $\rho_0\in L^1_+(\mathbb R^{n})\cap L^{m} (\mathbb R^{n})$, then the solution to the degenerate Keller-Segel with Bessel potential blows up in finite time.) In fact, Theorem \ref{thm} gives an exact classification of the initial data so that the solution either exists globally or blow-up in finite time. More precisely, it is the constant $(s^*)^{\frac{n-2}{2n(m-1)}}$, where $s^*$ is stated in (\ref{s*}), which classifies the initial data in $L^{\frac{2n}{n+2}}$ norm.
\end{rem}

\begin{rem} The exponents of $M_0$ in \eqref{F*} and \eqref{s*} are both negative due to the fact that $\frac{2n}{n+2}<m<2-\frac{2}{n}$. The assumption ${\mathcal F}(\rho_0)<{\mathcal F}^*$ in Theorem \ref{thm} gives a relation between the initial mass and the initial free energy, i.e.
\begin{eqnarray}\label{relation}
{\mathcal F}(\rho_0)M_0^{\frac{m(n+2)-2n}{2n-2-mn}}< \dfrac{2-\frac{2}{n}-m}{(m-1)(1-\frac{2}{n})}\Big(\frac{2n^2\alpha(n)}{C(n)}\Big)^{\frac{n(m-1)}{2n-2-m n}}.
\end{eqnarray}
As a conclusion, Theorem \ref{thm} implies that {\it the initial mass itself might not be an important quantity in the existence and blow up analysis in multi-dimension}. More precisely, no matter how small the initial mass is, the solution can still blow up in case that $\|\rho_0\|_{L^\frac{2n}{n+2}}>(s^*)^{\frac{n-2}{2n(m-1)}}$. No matter how large the initial mass is, there still exists a global weak solution if $\|\rho_0\|_{L^\frac{2n}{n+2}}<(s^*)^{\frac{n-2}{2n(m-1)}}$. The similar fact that the initial mass
is not a relevant quantity for blow-up in the multi-dimensional
Keller-Segel model is known in the literature, such as in \cite{CPZ} where (\ref{equs}) with $m=1$ was considered. Moreover, we can find a consistent phenomenon with this result in parabolic-parabolic model, such as in \cite{Winkler,CP}. In \cite{CP}, the norm of $\|\rho_0\|_{L^{\frac{n}{2}}}$ was used to discuss existence and blow-up. The author in \cite{Winkler} studied the case with smooth bounded domain with homogeneous Neumann boundary conditions, then obtained the existence result for small initial data in $L^q$, $q>\frac{n}{2}$ and if the domain is ball, there is always an unbounded solution developed from initial data with arbitrary small mass.

\begin{exam}\label{example1}
For given $\varepsilon_0>0$ arbitrarily small, let the initial data be
$$
\rho_0(x)=\left\{\begin{array}{ll} \varepsilon_0 \frac{K^n}{\alpha(n)}, & |x|\leq \frac{1}{K},\\ 0, & |x|>\frac{1}{K},\end{array}\right.
$$
where $K$ to be determined later.
Then
$$
\|\rho_0\|_{L^1}=\varepsilon_0, \quad\|\rho_0\|_{L^{\frac{2n}{n+2}}}=\varepsilon_0 \Big(\frac{K^n}{\alpha(n)}\Big)^{\frac{n-2}{2n}} \quad \mbox{ and }\quad \int_{\mathbb {R}^n}|x|^2\rho_0dx<\infty.
$$
Now we can choose $K$ large such that
\begin{eqnarray}\label{examrelmass}
\|\rho_0\|_{L^{\frac{2n}{n+2}}}>(s^*)^{\frac{n-2}{2n(m-1)}},
\end{eqnarray}
and
\begin{eqnarray}\label{examrelentropy}
{\mathcal F}(\rho_0)M_0^{\frac{m(n+2)-2n}{2n-2-mn}}< \dfrac{2-\frac{2}{n}-m}{(m-1)(1-\frac{2}{n})}\Big(\frac{2n^2\alpha(n)}{C(n)}\Big)^{\frac{n(m-1)}{2n-2-mn}}.
\end{eqnarray}
Therefore according to our result in theorem \ref{thm}, the solution must blow up in finite time.

We will give a detailed calculation of this example in the Appendix.

Similarly, we can find some initial data with large initial mass such that the solution exist globally.
\end{exam}

\end{rem}

It should also be mentioned that the constants appeared in the main result have close relation to the critical Hardy-Littlewood-Sobolev inequality. For completeness, we cite this result from  \cite{LiebAnalysis}.

\begin{prop} [H.-L.-S. inequality]\label{propHLS}
Let $\rho\in L^{\frac{2n}{n+2}}({\mathbb R^n})$, then
\begin{eqnarray}\label{HLSinequatlity}
\dint\dint_{\mathbb
R^n\times \mathbb
R^n}\dfrac{\rho(x)\rho(y)}{|x-y|^{n-2}} dx dy \leq C(n)
\|\rho\|_{L^{\frac{2n}{n+2}}}^2,
\end{eqnarray}
where
\begin{eqnarray} \label{Cn}
C(n)=\pi^{(n-2)/2}\dfrac{1}{\Gamma(n/2+1)}\left\{\dfrac{\Gamma(n/2)}{\Gamma(n)}\right\}^{-2/n}.
\end{eqnarray}
Moreover, the equality holds if and only
if $\rho(x)=A U_{\lambda,x_0}$, for some constant $A$ and parameters $\lambda>0$, $x_0\in \mathbb{R}^n$,
where
\begin{equation}\label{stationarysolutionu}
U_{\lambda,x_0}=2^{\frac{n+2}{4}}n^{\frac{n+2}{2}}
\left(\frac{\lambda}{\lambda^2+| x- x_0|^2}\right)^{\frac{n+2}{2}}.
\end{equation}
\end{prop}
This family of radially symmetric functions \eqref{stationarysolutionu} is also a class of stationary solution of the degenerate system \eqref{equs} with diffusion exponent $m=m_c=\frac{2n}{n+2}$. The readers are referred to \cite{CLW1} for the relations among stationary solution, Haddy-Littlewood-Sobolev inequality and conformal invariance of the free energy.
A direct scaling analysis tells us that $L^{\frac{2n}{n+2}}$ norm of $U_{\lambda, x_0}$ is a universal constant independent of the parameters $\lambda$ and $x_0$.

We can separate the free energy into two parts by using Haddy-Littlewood-Sobolev inequality \eqref{HLSinequatlity}, namely,
\begin{eqnarray*}
{\mathcal F}(\rho)&=&\dfrac{1}{m-1}\int_{\mathbb R^n}\rho^m( x,t)dx-\dfrac{C(n)}{2(n-2)n\alpha(n)}\|\rho\|^2_{L^{\frac{2n}{n+2}}}\\
&&+\dfrac{C(n)}{2(n-2)n\alpha(n)}\|\rho\|^2_{L^{\frac{2n}{n+2}}}-\dfrac{1}{2(n-2)n\alpha(n)}\int\int_{\mathbb R^n\times\mathbb R^n}\dfrac{\rho(x,t)\rho(y,t)}{|x-y|^{n-2}}d x dy\\
&=:& {\mathcal F}_1(\rho)+{\mathcal F}_2(\rho).
\end{eqnarray*}
Proposition \ref{propHLS} says that ${\mathcal F}_2(\rho)\geq 0$.

Since the first part of the free energy is concave in $L^\frac{2n}{n+2}$ norm of the solution, it is not difficult to get \emph{a priori} estimates, which shows that in the cases of supercritical and subcritical initial data, the quantity $\|\rho\|_{L^{\frac{2n}{n+2}}}$ can be bounded from below or from above separately. More precisely, if the initial free energy ${\mathcal F}(\rho_0)<{\mathcal F}^*$, then the following estimates hold
\begin{enumerate}
\item If $\|\rho_0\|_{\frac{2n}{n+2}}< (s^*)^{\frac{n-2}{2n(m-1)}}$, then  there exists a constant $\mu_1<1$ such that
$$
\|\rho(\cdot, t)\|_{\frac{2n}{n+2}}<(\mu_1s^*)^{\frac{n-2}{2n(m-1)}}, \mbox{ for all }t>0.
$$
\item If $\|\rho_0\|_{\frac{2n}{n+2}}> (s^*)^{\frac{n-2}{2n(m-1)}}$, then there exists a constant $\mu_2>1$ such that
$$
\|\rho(\cdot, t)\|_{\frac{2n}{n+2}}> (\mu_2s^*)^{\frac{n-2}{2n(m-1)}}, \mbox{ for all } t>0.
$$
\end{enumerate}
We will give the proof of the first fact for the regularized solution in the Lemma \ref{upbounded} in section 2, and show that the second is true in Lemma \ref{lowbounded} in section 3.

This paper is arranged as follows. In section 2, we will give the proof of the global existence of weak solution. After introducing the regularized problem, a uniform estimate for the $L^{\frac{2n}{n+2}}$ norm of the regularized solution by using decomposition of the free energy is obtained. Based on this estimate, further estimates, including the spacial and time derivatives, are derived. Then the global existence follows from standard compactness arguments with the help of Aubin's lemma. In section 3, with supercritical initial data, it is shown that any solution will blow-up in finite time by studying the time derivative of second moment.

\section{Existence of weak solution}
We follow the same way on the construction of the regularized problem as in \cite{S1,SK,BCL09}, namely,
\begin{align}
\left\{
\begin{array}{llll}\label{regularized equa}\smallskip
&\partial_t \rho_\varepsilon=\Delta [(\rho_\varepsilon+\varepsilon)^m-\varepsilon^m]
-{\rm div}((\rho_\varepsilon+\varepsilon)\nabla c_\varepsilon),&& x\in \mathbb R^n ,~t\geq 0,\\
&-\Delta c_\varepsilon=J_\varepsilon * \rho_\varepsilon,
&& x\in \mathbb R^n,~t\geq 0,\\
&\rho(\bx,0)=\rho_{0\varepsilon}(\bx), && x\in \mathbb R^n
\end{array}
\right.
\end{align}
for $\varepsilon>0$, $J_\varepsilon(x)=\frac{1}{\varepsilon^n}J(\frac{x}{\varepsilon})$, $J(x)=\frac{1}{\alpha(n)}(1+|x|^2)^{-(n+2)/2}$ satisfying $\dint_{\mathbb{R}^n}J_{\varepsilon}(x) dx=1$. A simple computation derives
 $$
 c_\varepsilon=\dfrac{1}{n(n-2)\alpha(n)} \dint_{\mathbb{R}^n} \dfrac{1}{(|x-y|^2+\varepsilon^2)^{\frac{n-2}{2}}}
\rho_\varepsilon(y)dy.
 $$

The initial data $\rho_{0\varepsilon}$ is the regularization of the function $\rho_0$, it satisfies that there exists a positive constant $\delta$ such that for all $0<\varepsilon<\delta$,
 \begin{align*}
&\rho_{0\varepsilon}>0,\qquad \rho_{0\varepsilon}\in L^r(\mathbb{R}^n),~~ r\geq 1, \qquad \|\rho_{0\varepsilon}\|_{L^1}=\|\rho_0\|_{L^1}=M_0,
\end{align*}
Moreover, as $\varepsilon\rightarrow 0$,
\begin{align*}
&\int_{\mathbb{R}^n}|x|^2\rho_{0\varepsilon}dx\rightarrow\int_{\mathbb{R}^n}|x|^2\rho_{0}dx, \qquad {\mathcal F}_\varepsilon(\rho_{0\varepsilon})\rightarrow {\mathcal F}(\rho_0)\\
&\mbox{If}~~ \rho_0\in L^{p} \mbox{ for some } p>1, ~~\mbox{then}~~\|\rho_{0\varepsilon}-\rho_0\|_{L^p}\rightarrow0, \mbox{ as } ~~\varepsilon\rightarrow 0,
\end{align*}
where ${\mathcal F}_\varepsilon(\rho_{0\varepsilon})$ is the initial regularized entropy, see \eqref{freeenergy}.

The classical parabolic theory implies that the above regularized problem has a global smooth nonnegative solution $\rho_\varepsilon$ for $t>0$ if the initial data is nonnegative. Notice that the solution of the regularized problem (\ref{regularized equa}) still conserves the mass.

We will mainly focus on the estimates of the regularized solutions in this section. After getting $L^\frac{2n}{n+2}$ estimate with the help of the free energy, we obtain the uniform $L^p$ estimates by using standard method. Furthermore, the uniform estimates for space and time derivatives will be derived carefully. With all these uniform estimates, a standard compactness argument as in \cite{CLW1,BL} by using Aubin's lemma will give the global existence.

From now on, we will present the uniform estimates in five steps and will skip the compactness arguments.

{\bf Step 1.} Free energy of the regularized problem\\
The free energy on the regularized solution $\rho_\varepsilon$ is
\begin{equation}\label{freeenergy}
{\mathcal F}_{\varepsilon}(\rho_\varepsilon)=\dfrac{1}{m-1}\int_{\mathbb R^n}((\rho_\varepsilon+\varepsilon)^m-\varepsilon^m)dx - \dfrac{1}{2}\int_{\mathbb R^n} \rho_{\varepsilon}c_{\varepsilon}d x.
\end{equation}
Or, the free energy has an equivalent form in the following
\begin{eqnarray}\label{freeenergy2}
{\mathcal F}_{\varepsilon}(\rho_\varepsilon)&=&\dfrac{1}{m-1}\int_{\mathbb R^n}((\rho_\varepsilon+\varepsilon)^m-\varepsilon^m)dx \nonumber\\ &&-\dfrac{1}{2(n-2)n\alpha(n)}\int\int_{\mathbb R^n\times\mathbb R^n}\dfrac{\rho_{\varepsilon}(x,t)\rho_{\varepsilon}(y,t)}{(|x-y|^2+\varepsilon^2)^{\frac{n-2}{2}}}d x dy.
\end{eqnarray}
It is easy to check that ${\mathcal F}_{\varepsilon}(\rho_\varepsilon)$ is non-increasing in time. In fact, the system (\ref{regularized equa}) has the gradient flow structure
\begin{equation}\label{gdientflux}
 \rho_{\varepsilon t}  = {\rm div}\left((\rho_{\varepsilon}+\varepsilon) \nabla\left(\frac{m}{m-1}(\rho_{\varepsilon}+\varepsilon)^{m-1}-c_{\varepsilon}\right)\right).
\end{equation}
Now by taking $\frac{m}{m-1}\left((\rho_{\varepsilon}+\varepsilon)^{m-1}-\varepsilon^{m-1}\right)-c_{\varepsilon}$ as a test function, we have the following entropy-entropy production relation
\begin{eqnarray*}
\dfrac{d}{dt} \mathcal{F}_{\varepsilon} (\rho_{\varepsilon}(\cdot, t)) +\dint_{\mathbb{R}^n}( \rho_{\varepsilon} +\varepsilon) \Big|\nabla\left(\dfrac{m}{m-1}(\rho_{\varepsilon} +\varepsilon)^{m-1}-c_\varepsilon\right)\Big|^2 dx =0.
\end{eqnarray*}
The monotone decreasing property of the free energy follows immediately by the nonnegativity of the entropy production.


Next, we separate the free energy into two parts by using Haddy-Littlewood-Sobolev inequality \eqref{HLSinequatlity}, i.e.,
\begin{eqnarray*}
{\mathcal F}_{\varepsilon}(\rho_{\varepsilon})&=&\dfrac{1}{m-1}\int_{\mathbb R^n}\left((\rho_{\varepsilon}+{\varepsilon})^m-{\varepsilon}^m\right)dx-\dfrac{C(n)}{2(n-2)n\alpha(n)}\|\rho_{\varepsilon}\|^2_{L^{\frac{2n}{n+2}}}\\
&&+\dfrac{C(n)}{2(n-2)n\alpha(n)}\|\rho_{\varepsilon}\|^2_{L^{\frac{2n}{n+2}}}-\dfrac{1}{2(n-2)n\alpha(n)}\int\int_{\mathbb R^n\times\mathbb R^n}\dfrac{\rho_{\varepsilon}(x,t)\rho_{\varepsilon}(y,t)}{(|x-y|^2+{\varepsilon}^2)^{\frac{n-2}{2}}}d x dy\\
&\geq& \dfrac{1}{m-1}\int_{\mathbb R^n}\rho_{\varepsilon}^m dx-\dfrac{C(n)}{2(n-2)n\alpha(n)}\|\rho_{\varepsilon}\|^2_{L^{\frac{2n}{n+2}}}\\
&&+\dfrac{C(n)}{2(n-2)n\alpha(n)}\|\rho_{\varepsilon}\|^2_{L^{\frac{2n}{n+2}}}-\dfrac{1}{2(n-2)n\alpha(n)}\int\int_{\mathbb R^n\times\mathbb R^n}\dfrac{\rho_{\varepsilon}(x,t)\rho_{\varepsilon}(y,t)}{|x-y|^{n-2}}d x dy\\
&=:& {\mathcal F}_1(\rho_{\varepsilon})+{\mathcal F}_2(\rho_{\varepsilon}).
\end{eqnarray*}
Proposition \ref{propHLS} shows that the second part of the free energy is nonnegative, i.e. ${\mathcal F}_2(\rho_{\varepsilon})\geq 0$.

Due to $m>\dfrac{2n}{n+2}$, interpolation shows that
\begin{eqnarray}
\|\rho_{\varepsilon}\|_{L^{\frac{2n}{n+2}}}\leq \|\rho_{\varepsilon}\|^{1-\theta}_{L^1}\|\rho_{\varepsilon}\|^{\theta}_{L^{m}}, \quad \theta=\frac{m(n-2)}{2n(m-1)}. \label{interpolation}
\end{eqnarray}

Thus the first part of the free energy is
\begin{eqnarray}\label{freeenergy1}
{\mathcal F}_1(\rho_{\varepsilon})&=&\dfrac{1}{m-1}\int_{\mathbb R^n}\rho_{\varepsilon}^m( x,t)dx-\dfrac{C(n)}{2(n-2)n\alpha(n)}\|\rho_{\varepsilon}\|^2_{L^{\frac{2n}{n+2}}}\nonumber\\
&\geq& \dfrac{1}{m-1}\|\rho_{\varepsilon}\|^{\frac{(\theta-1)m}{\theta}}_{L^1}\|\rho_{\varepsilon}\|^{\frac{m}{\theta}}_{L^{\frac{2n}{n+2}}}
-\dfrac{C(n)}{2(n-2)n\alpha(n)}\|\rho_{\varepsilon}\|^2_{L^{\frac{2n}{n+2}}}\\
&\geq& \dfrac{1}{m-1}M_0^{\frac{2n-m(n+2)}{n-2}}\|\rho_{\varepsilon}\|^{\frac{2n(m-1)}{n-2}}_{L^{\frac{2n}{n+2}}}
-\dfrac{C(n)}{2(n-2)n\alpha(n)}\|\rho_{\varepsilon}\|^2_{L^{\frac{2n}{n+2}}}. \nonumber
\end{eqnarray}

According to the previous analysis, let
$$
f(s)=\dfrac{1}{m-1}M_0^{\frac{2n-m(n+2)}{n-2}}s-\dfrac{C(n)}{2(n-2)n\alpha(n)}s^{\frac{n-2}{n(m-1)}}.
$$
We now have a lower bound of the first part of free energy, i.e. $f\Big(\|\rho_{\varepsilon}\|^{\frac{2n(m-1)}{n-2}}_{L^{\frac{2n}{n+2}}}\Big)\leq{\mathcal F}_1(\rho_{\varepsilon})$.

{\bf Step 2.} Uniform $L^{\frac{2n}{n+2}}$ norm estimate of the regularized solution.

 The following lemma shows that for subcritical initial data, the quantity $\|\rho_{\varepsilon}\|_{L^{\frac{2n}{n+2}}}$ can be bounded.
\begin{lem} \label{upbounded}
If the initial free energy ${\mathcal F}_{\varepsilon}(\rho_{0\varepsilon})<{\mathcal F}^*:=f(s^*)$, $\|\rho_{0\varepsilon}\|_{\frac{2n}{n+2}}< (s^*)^{\frac{n-2}{2n(m-1)}}$, let $\rho_{\varepsilon}(x,t)$ be a solution of problem \eqref{regularized equa}, then  there exists a constant $\mu_1<1$ such that
$$
\|\rho_{\varepsilon}(\cdot, t)\|_{\frac{2n}{n+2}}<(\mu_1s^*)^{\frac{n-2}{2n(m-1)}}, \mbox{ for all }t>0,
$$
where $s^*$ is the maximum point of $f(s)$:
\begin{eqnarray}
s^*=\Big(\dfrac{2n^2\alpha(n)M_0^{\frac{2n-m(n+2)}{n-2}}}{C(n)}\Big)^{\frac{n(m-1)}{2n-2-mn}}.\label{defs*}
\end{eqnarray}
\end{lem}
\begin{proof}
Notice that $1<m<2-\frac{2}{n}$ implies $\frac{n-2}{n(m-1)}>1$, we know that $f(s)$ is a strictly concave function in $0<s<\infty$. Directly calculation shows that
$$
f'(s)=\dfrac{1}{m-1}M_0^{\frac{2n-m(n+2)}{n-2}}-\dfrac{C(n)}{2(n-2)n\alpha(n)}\frac{n-2}{n(m-1)}s^{\frac{2n-2-mn}{n(m-1)}}.
$$
As a consequence, $s^*$ is a unique maximum point of $f(s)$.
Therefore the important property of $f$ is that $f(s)$ is monotone increasing for $0<s<s^*$, while $f(s)$ is monotone decreasing for $s>s^*$.

In the case that initial free energy ${\mathcal F}_{\varepsilon}(\rho_{0\varepsilon})<f(s^*)$, we can make it even smaller, i.e. there is a $\delta<1$ such that ${\mathcal F}_{\varepsilon}(\rho_{0\varepsilon})<\delta f(s^*)$.

Combining all the facts we know, including the interpolation, Haddy-Littlewood-Sobolev inequality and the monotonicity of free energy, we have
\begin{eqnarray}\label{leq}
f\Big(\|\rho_{\varepsilon}\|^{\frac{2n(m-1)}{n-2}}_{L^{\frac{2n}{n+2}}}\Big)\leq{\mathcal F}_1(\rho_{\varepsilon})\leq {\mathcal F}_{\varepsilon}(\rho_{\varepsilon})\leq {\mathcal F}_{\varepsilon}(\rho_{0\varepsilon})<\delta f(s^*).
\end{eqnarray}
If initially $\|\rho_{0\varepsilon}\|^{\frac{2n(m-1)}{n-2}}_{L^{\frac{2n}{n+2}}}<s^*$, due to the fact that $f(s)$ is increasing in $0<s<s^*$, there exists a $\mu_1<1$ such that  $\|\rho_{\varepsilon}\|^{\frac{2n(m-1)}{n-2}}_{L^{\frac{2n}{n+2}}}<\mu_1 s^*$.
\end{proof}

{\bf Step 3.} Uniform $L^p$ $(1<p<n)$ estimates of the regularized solution.
$\newline$
 Under the assumption of $\rho_{0\varepsilon}\in L^{p}(\mathbb R ^n)$ with $1<p<n$, we will give the estimate on $\|\rho_{\varepsilon}\|_{L^p}$, and as a byproduct, also the uniform estimates on space derivatives $\nabla \rho_{\varepsilon}^{\frac{m+p-1}{2}}$ and $\nabla c_{\varepsilon}$.

\begin{lem} \label{lemLp}
Assume $\rho_{0\varepsilon}\in L^1(\mathbb R ^n)\cap L^{p}(\mathbb R ^n)$, $\|\rho_{0\varepsilon}\|_{L^{\frac{2n}{n+2}}} < (s^*)^{\frac{n-2}{2n(m-1)}}$ and ${\mathcal F}_{\varepsilon}(\rho_{0\varepsilon})<{\mathcal F}^*:=f(s^*)$, $\rho_{\varepsilon}$ is a smooth solution of the regularized problem \eqref{regularized equa}, then
\begin{eqnarray}
&&\|\rho_{\varepsilon}\|_{L^\infty(0,T; L^p(\mathbb{R}^n)\cap L^{p+1}(0,T; L^{p+1}(\mathbb{R}^n)))}\leq C, \quad
\|\nabla \rho_{\varepsilon}^{\frac{m+p-1}{2}}\|_{L^2(0,T; L^2(\mathbb{R}^n))}\leq C, \label{est2324}
\end{eqnarray}
moreover, for $1<p<n$, it holds
\begin{eqnarray}
&& \|\nabla c_{\varepsilon}\|_{L^\infty(0,T; L^s(\mathbb{R}^n))}\leq C, \quad s\in \big(\dfrac{n}{n-1},\dfrac{np}{n-p}\big]. \label{est25}
\end{eqnarray}
\end{lem}

\begin{proof}
Multiplying the first equation of (\ref{regularized equa}) by $p\rho_{\varepsilon}^{p-1}$ with $p>1$, we have
\begin{eqnarray*}
\dfrac{d}{dt}\dint_{\mathbb{R}^n} \rho_{\varepsilon}^p dx &=& -pm(p-1)\dint_{\mathbb{R}^n} (\rho_{\varepsilon}+\varepsilon)^{m-1}\rho_{\varepsilon}^{p-2} |\nabla\rho_{\varepsilon}|^2 dx +(p-1) \dint_{\mathbb{R}^n} \nabla \rho_{\varepsilon}^p\cdot \nabla c_{\varepsilon} dx\\
&&+\varepsilon p\int_{\mathbb{R}^n}\nabla \rho_{\varepsilon}^{p-1}\cdot \nabla c_{\varepsilon} dx \\
&\leq& -pm(p-1)\dint_{\mathbb{R}^n} \rho_{\varepsilon}^{p+m-3} |\nabla\rho_{\varepsilon}|^2 dx +(p-1)\dint_{\mathbb{R}^n}\rho_{\varepsilon}^{p+1}dx+p\varepsilon\dint_{\mathbb{R}^n}\rho_{\varepsilon}^{p}dx\\
&=& -\dfrac{4pm(p-1)}{(m+p-1)^2}\dint_{\mathbb{R}^n}|\nabla\rho_{\varepsilon}^{\frac{m+p-1}{2}}|^2 dx +(p-1)\dint_{\mathbb{R}^n}\rho_{\varepsilon}^{p+1}+p\varepsilon\dint_{\mathbb{R}^n}\rho_{\varepsilon}^{p}dx.
\end{eqnarray*}
Now we will focus on the estimate on $\dint_{\mathbb{R}^n}\rho_{\varepsilon}^{p+1}$.
\begin{eqnarray*}
\dint_{\mathbb{R}^n}\rho_{\varepsilon}^{p+1} &=&\Big\|\rho_{\varepsilon}^{\frac{m+p-1}{2}}\Big\|_{L^{\frac{2(p+1)}{m+p-1}}}^{\frac{2(p+1)}{m+p-1}}\\
&\leq & G^{\frac{2(p+1)}{m+p-1}} \Big\|\nabla \rho_{\varepsilon}^{\frac{m+p-1}{2}}\Big\|_{L^2}^{\alpha\frac{2(p+1)}{m+p-1}} \cdot \Big\|\rho_{\varepsilon}^{\frac{m+p-1}{2}}\Big\|_{L^r}^{(1-\alpha)\frac{2(p+1)}{m+p-1}},
\end{eqnarray*}
where $G$ is the constant from Gagliardo-Nirenberg-Sobolev ineqality,
$$
\dfrac{m+p-1}{2}r=\dfrac{2n}{n+2},\quad\quad \dfrac{m+p-1}{2(p+1)}=\dfrac{\alpha (n-2)}{2n} +\dfrac{1-\alpha}{r},
$$
and
$$
\alpha=\dfrac{\frac{m+p-1}{2}(\frac{n+2}{2n}-\frac{1}{p+1})}{\frac{(n+2)(m+p-1)-2(n-2)}{4n}}.
$$
In the next, we will use notation
$$
\nu:=\alpha\dfrac{2(p+1)}{m+p-1}=\frac{2(n+2)(p+1)-4n}{(n+2)(m+p-1)-2(n-2)} <2 \quad \mbox{ in the case of }\quad m>\dfrac{2n}{n+2}.
$$
Thus by Young's inequality, we get
\begin{eqnarray}\label{lp+1}
\dint_{\mathbb{R}^n} \rho_{\varepsilon}^{p+1}&\leq &  G^{\frac{2(p+1)}{m+p-1}} \Big\|\nabla \rho_{\varepsilon}^{\frac{m+p-1}{2}}\Big\|_{L^2}^{\nu} \|\rho_{\varepsilon}\|_{L^{\frac{2n}{n+2}}}^{(1-\alpha)(p+1)}\nonumber\\
&\leq &  G^{\frac{2(p+1)}{m+p-1}} \Big(\epsilon \Big\|\nabla \rho_{\varepsilon}^{\frac{m+p-1}{2}}\Big\|_{L^2}^2 +C(\epsilon)\|\rho_{\varepsilon}\|_{L^{\frac{2n}{n+2}}}^{\frac{2(1-\alpha)(p+1)}{2-\nu}} \Big).
\end{eqnarray}
Now we can choose $\epsilon$ such that
$$
(p-1)G^{\frac{2(p+1)}{m+p-1}}\epsilon=\dfrac{2pm(p-1)}{(m+p-1)^2}.
$$
By using the boundedness of $\|\rho_{\varepsilon}\|_{L^{\frac{2n}{n+2}}}$ from Lemma \ref{upbounded}, we have
\begin{eqnarray}\label{qq}
\dfrac{d}{dt}\dint_{\mathbb{R}^n}\rho_{\varepsilon}^p dx +\dfrac{2pm(p-1)}{(m+p-1)^2} \dint_{\mathbb{R}^n}|\nabla \rho_{\varepsilon}^{\frac{m+p-1}{2}}|^2 dx \leq p\varepsilon\dint_{\mathbb{R}^n}\rho_{\varepsilon}^{p}dx+C(M_0,p,n).
\end{eqnarray}
Gronwall's inequality implies that $\rho_{\varepsilon}\in L^{\infty}(0,T;L^p(\mathbb{R}^n))$.
Therefore we have the uniform estimate by integrating (\ref{qq}) in $t$, for any fixed $T>0$,
\begin{eqnarray*}
\sup_{0\leq t\leq T}\int_{\mathbb{R}^n}\rho_{\varepsilon}^p(x,t) dx +\frac{2pm(p-1)}{(m+p-1)^2}\int^T_0\dint_{\mathbb{R}^n}|\nabla \rho_{\varepsilon}^{\frac{m+p-1}{2}}|^2 dx dt \leq C(M_0,p,n,T).
\end{eqnarray*}
Moreover combining this estimate with (\ref{lp+1}), it is easy to see that $\rho_{\varepsilon} $ $\in $ $L^{p+1}(0,T; L^{p+1}$ $(\mathbb R^n))$. The estimate for $\nabla c_{\varepsilon}$ in \eqref{est25} can be directly obtained from weak Young's inequality.
\end{proof}
\begin{rem}
The above lemma gives a general $L^p$ estimate. In particular, we can take $p=m$ and get the estimate $\rho_{\varepsilon}\in L^{\infty}(0,T;L^{m}({\mathbb{R}^n}))\cap L^{m+1}(0,T;L^{m+1}({\mathbb{R}^n}))$ which will be used later.
\end{rem}
\begin{rem}
The fact that $m>\frac{2n}{n+2}$ is very important in the above proof. It makes the use of Young's inequality successful (see (\ref{lp+1})), which is impossible in the case $m=\frac{2n}{n+2}$, $\nu=2$.
\end{rem}

{\bf Step 4.} Uniform estimates for the space derivatives

The estimate on space derivative of $\rho_{\varepsilon}$ is important in order to use Aubin's lemma for compactness arguments. We will use the $L^p$ estimate when $p=m$.

\begin{lem}\label{nablarho} Assume $p=m$ and the assumptions of Lemma \ref{upbounded} hold, then
\begin{eqnarray}
&&\|\nabla\rho_{\varepsilon}\|_{L^{2}(0,T; L^{\frac{2m}{3-m}}(\mathbb{R}^n))}\leq C, \quad \mbox{ in the case of } m< \frac{3}{2}, \label{est28}\\
&&\|\nabla\rho_{\varepsilon}\|_{L^{2}(0,T; L^{2}(\mathbb{R}^n))}\leq C, \quad \mbox{ in the case of } m\geq \frac{3}{2}.\label{est29}
\end{eqnarray}
\end{lem}

\begin{proof}
In the case of $m<\frac{3}{2}$, using (\ref{est2324}), it holds for $p=m$ that
 \begin{eqnarray}
&&\|\rho_{\varepsilon}\|_{L^\infty(0,T; L^m(\mathbb{R}^n))}\leq C, \quad \|\nabla \rho_{\varepsilon}^{m-\frac{1}{2}}\|_{L^2(0,T; L^2(\mathbb{R}^n))}\leq C. \label{estp=m}
\end{eqnarray}
We can use the expression
\begin{eqnarray*}
\nabla\rho_{\varepsilon}=\dfrac{2}{2m-1}\rho_{\varepsilon}^{\frac{3}{2}-m} \nabla\rho_{\varepsilon}^{m-\frac{1}{2}},
\end{eqnarray*}
then H\"older inequality and \eqref{estp=m} imply \eqref{est28}.

In the case of $m\geq\frac{3}{2}$, taking $\rho_{\varepsilon}^{2-m}$ as test function in (\ref{equs}), we have
\begin{eqnarray*}
&&\frac{1}{3-m}\frac{d}{dt}\int_{ {\mathbb R}^n} \rho_{\varepsilon}^{3-m} d x +m(2-m)
\int_{ {\mathbb R}^n} \left|\nabla \rho_{\varepsilon}\right|^2 d x
\leq\frac{2-m}{3-m}\int_{ {\mathbb R}^n} \rho_{\varepsilon}^{4-m}d x+\varepsilon\int_{\mathbb{R}^n}\rho_{\varepsilon}^{3-m}dx.
\end{eqnarray*}
Next we only need to estimate $\int_{ {\mathbb R}^n} \rho_{\varepsilon}^{4-m}d x$ by $\|\rho_{\varepsilon}\|_{L^m}$ and $\|\nabla \rho_{\varepsilon}^{m-\frac{1}{2}}\|_{L^2(0,T;L^2(\mathbb{R}^n))}$.
By Gagliardo-Nirenberg-Sobolev inequality, we have
 \begin{eqnarray}
 \int_{ {\mathbb R}^n} \rho_{\varepsilon}^{4-m}d x
 =\|\rho_{\varepsilon}^{m-1/2}\|_{L^{\frac{4-m}{m-1/2}}}^{\frac{4-m}{m-1/2}}&\leq& C\|\nabla \rho_{\varepsilon}^{m-1/2}\|_{L^2}^{\theta\frac{4-m}{m-1/2}}\|\rho_{\varepsilon}^{m-1/2}\|_{L^{\frac{m}{m-1/2}}}^{(1-\theta)\frac{4-m}{m-1/2}}\nonumber\\
 &=&C\|\nabla\rho_{\varepsilon}^{m-1/2}\|_{L^2}^{\theta\frac{4-m}{m-1/2}}\|\rho_{\varepsilon}\|_{L^m}^{(1-\theta)(4-m)},\label{estL4-m}
\end{eqnarray}
where $0<\theta=\frac{2(2-m)(m-1/2)}{m(4-m)\left(\frac{m-1/2}{m}-\frac{n-2}{2n}\right)}<1$.
Thus it remains to show if $m\geq \frac{3}{2}$ and $\frac{2n}{n+2}<m<2-\frac{2}{n}$, it holds that
\begin{eqnarray}
\theta\frac{4-m}{m-1/2}=\frac{2(2-m)}{m-\frac{1}{2}-\frac{m(n-2)}{2n}}\leq 2.\label{exponent}
\end{eqnarray}
Actually, (\ref{exponent}) is equivalent to $m\geq \frac{5n}{3n+2}$, which can be obtained from the following two facts.
\begin{itemize}
\item When $n\geq 6$, since $\frac{2n}{n+2}\geq \frac{5n}{3n+2}$, we have $m>\frac{5n}{3n+2}$;
\item When $n<6$, since $\frac{3}{2}> \frac{5n}{3n+2}$, we have $m>\frac{5n}{3n+2}$.
\end{itemize}

Now by integrating (\ref{estL4-m}) in time, we have
\begin{eqnarray*}
 \int_0^T\int_{ {\mathbb R}^n} \rho_{\varepsilon}^{4-m}d x dt \leq
 C \left(\|\rho_{\varepsilon}\|_{L^\infty( 0,T; L^{m}(\mathbb R^n))},\|\nabla \rho_{\varepsilon}^{m-1/2}\|_{L^2 (0,T; L^2(\mathbb R^n))}, T\right).
\end{eqnarray*}

Therefore,
\begin{eqnarray*}
&&\frac{1}{(3-m)}\int_{ {\mathbb R}^n} \rho_{\varepsilon}^{3-m} d x +m(2-m)
\int^T_0\int_{ {\mathbb R}^n} \left|\nabla \rho_{\varepsilon} \right|^2 d x dt\\
&\leq &\frac{1}{(3-m)}\|\rho_{0\varepsilon}\|^{3-m}_{L^{3-m}}+C\leq C\left(\|\rho_{0\varepsilon}\|_{L^m},\|\rho_{0\varepsilon}\|_{L^1}\right)+C,
\end{eqnarray*}
where we have used the fact that $3-m\leq m$. So, (\ref{est29}) holds.
\end{proof}

{\bf Step 5.} Uniform estimate for the time derivative.

This subsection will give another important fact in order to use Aubins lemma, i.e. the estimate of the time derivative of $\rho_{\varepsilon}$.

\begin{lem}\label{time derivative}
Assume $p=m$ and the assumptions of Lemma \ref{upbounded} hold, then
\begin{eqnarray*}
\|\partial_t\rho_{\varepsilon}\|_{L^2(0,T; W_{\rm loc}^{-1,s}(\mathbb{R}^n))}\leq C, \quad s=\min\{\dfrac{2m}{m+1},\dfrac{nm(m+1)}{nm+(n-m)(m+1)}\}>1.
\end{eqnarray*}
\end{lem}

\begin{proof}
By using the weak formulation of the equation, we know the estimate for time derivative $\partial_t \rho_\varepsilon$ can be obtained directly from the estimates on $\nabla(\rho_{\varepsilon}+\varepsilon)^m$ and $(\rho_{\varepsilon}+{\varepsilon})\cdot\nabla c_{\varepsilon}$. We will prove the following facts,
\begin{eqnarray*}
&&\|\nabla(\rho_{\varepsilon}+\varepsilon)^m\|_{L^2(0,T; L^{\frac{2m}{m+1}}(\mathbb{R}^n))}\leq C,\\
&&\|(\rho_{\varepsilon}+{\varepsilon})\cdot\nabla c_{\varepsilon}\|_{L^{m+1}(0,T; L^{\frac{nm(m+1)}{nm+(n-m)(m+1)}}(\mathbb{R}^n))}\leq C.
\end{eqnarray*}

In fact,
\begin{eqnarray}\label{b1}
|\nabla(\rho_{\varepsilon}+\varepsilon)^m|&=&m|(\rho_{\varepsilon}+\varepsilon)^{m-1}|\cdot|\nabla \rho_{\varepsilon}|\nonumber\\
&\leq& m|(\rho_{\varepsilon}^{m-1}+\varepsilon^{m-1})|\cdot|\nabla \rho_{\varepsilon}|\leq  |\nabla\rho_{\varepsilon}^{m}|+m\varepsilon^{m-1}|\nabla \rho_{\varepsilon}|.
\end{eqnarray}
By writing
$$
|\nabla\rho_{\varepsilon}^m|=\big|\dfrac{2m}{2m-1}\rho_{\varepsilon}^{1/2}\nabla \rho_{\varepsilon}^{m-1/2}\big|,
$$
H\"older inequality and lemma \ref{lemLp}, we have
\begin{eqnarray*}
\dint_{\mathbb{R}^n}|\nabla\rho_{\varepsilon}^m|^{\frac{2m}{m+1}}\leq C \Big(\dint_{\mathbb{R}^n}\rho_{\varepsilon}^m\Big)^{\frac{1}{m+1}} \Big(\dint_{\mathbb{R}^n}|\nabla\rho_{\varepsilon}^{m-1/2}|^2\Big)^{\frac{m}{m+1}}.
\end{eqnarray*}
Therefore,
\begin{eqnarray*}
\dint^T_0 \|\nabla\rho_{\varepsilon}^m\|_{L^{\frac{2m}{m+1}}}^2\leq \dint^T_0 \|\rho_{\varepsilon}\|_{L^m}\|\nabla\rho_{\varepsilon}^{m-1/2}\|_{L^2}^2 dt \leq C,
\end{eqnarray*}
i.e.,
\begin{eqnarray}
\|\nabla \rho_{\varepsilon}^m\|_{L^2(0,T; L^{\frac{2m}{m+1}}(\mathbb{R}^n))}\leq C.\label{est26}
\end{eqnarray}
By Lemma \ref{nablarho}, since $\frac{2m}{m+1}<\min\{2,\frac{2m}{3-m}\}$ and (\ref{est26}), we know that
\begin{eqnarray*}
\nabla (\rho_{\varepsilon}+\varepsilon)^m\in L^2(0,T; L_{\rm loc}^{ \frac{2m}{m+1}}(\mathbb{R}^n)).
\end{eqnarray*}
As a direct consequence of Lemma \ref{lemLp}, we have
\begin{eqnarray}
\|\rho_{\varepsilon}\cdot\nabla c_{\varepsilon}\|_{L^{m+1}(0,T; L^{\frac{nm(m+1)}{nm+(n-m)(m+1)}}(\mathbb{R}^n))}\leq C,\label{est27}
\end{eqnarray}
where $\frac{nm(m+1)}{nm+(n-m)(m+1)}>1$ due to $\frac{2n}{n+2}<m<2-\frac{2}{n}$.
By Lemma \ref{lemLp} with (\ref{est27}) and noticing $\frac{nm(m+1)}{nm+(n-m)(m+1)}\in (\frac{n}{n-1},\frac{mn}{n-m}]$, we get
\begin{eqnarray*}
\|(\rho_{\varepsilon}+{\varepsilon})\cdot\nabla c_{\varepsilon}\|_{L^{m+1}(0,T; L^{\frac{nm(m+1)}{nm+(n-m)(m+1)}}(\mathbb{R}^n))}\leq C.
\end{eqnarray*}
\end{proof}

\section{Blow up of the solution}
In this section, we will discuss the blow-up of the solution when $\|\rho_0\|_{L^{\frac{2n}{n+2}}}>(s^*)^{\frac{n-2}{2n(m-1)}}$ and ${\mathcal F}(\rho_0)< {\mathcal F}^*:=f(s^*)$. Before we prove the result of blow-up, we need to give a key lemma that shows in the cases of subcritical initial data, the quantity $\|\rho\|_{L^{\frac{2n}{n+2}}}$ can be bounded from below.

\subsection{Lower bound of $\|\rho\|_{L^{\frac{2n}{n+2}}}$}
$\newline$
Similar to the decomposition of free energy of the regularized problem,
we can separate the free energy into two parts by using Haddy-Littlewood-Sobolev inequality \eqref{HLSinequatlity}
\begin{eqnarray*}
{\mathcal F}(\rho)&=&\dfrac{1}{m-1}\int_{\mathbb R^n}\rho^m( x,t)dx-\dfrac{C(n)}{2(n-2)n\alpha(n)}\|\rho\|^2_{L^{\frac{2n}{n+2}}}\\
&&+\dfrac{C(n)}{2(n-2)n\alpha(n)}\|\rho\|^2_{L^{\frac{2n}{n+2}}}-\dfrac{1}{2(n-2)n\alpha(n)}\int\int_{\mathbb R^n\times\mathbb R^n}\dfrac{\rho(x,t)\rho(y,t)}{|x-y|^{n-2}}d x dy\\
&=:& {\mathcal F}_1(\rho)+{\mathcal F}_2(\rho).
\end{eqnarray*}
Proposition \ref{propHLS} says that that ${\mathcal F}_2(\rho)\geq 0$.

Due to $m>\dfrac{2n}{n+2}$, interpolation tells us
\begin{eqnarray*}
\|\rho\|_{L^{\frac{2n}{n+2}}}\leq \|\rho\|^{1-\theta}_{L^1}\|\rho\|^{\theta}_{L^{m}}, \quad \theta=\frac{m(n-2)}{2n(m-1)}.
\end{eqnarray*}

Thus the first part of the free energy is
\begin{eqnarray*}
{\mathcal F}_1(\rho)&=&\dfrac{1}{m-1}\int_{\mathbb R^n}\rho^m( x,t)dx-\dfrac{C(n)}{2(n-2)n\alpha(n)}\|\rho\|^2_{L^{\frac{2n}{n+2}}}\nonumber\\
&\geq& \dfrac{1}{m-1}\|\rho\|^{\frac{(\theta-1)m}{\theta}}_{L^1}\|\rho\|^{\frac{m}{\theta}}_{L^{\frac{2n}{n+2}}}
-\dfrac{C(n)}{2(n-2)n\alpha(n)}\|\rho\|^2_{L^{\frac{2n}{n+2}}}\\
\end{eqnarray*}
\begin{eqnarray*}
&\geq& \dfrac{1}{m-1}M_0^{\frac{2n-m(n+2)}{n-2}}\|\rho\|^{\frac{2n(m-1)}{n-2}}_{L^{\frac{2n}{n+2}}}
-\dfrac{C(n)}{2(n-2)n\alpha(n)}\|\rho\|^2_{L^{\frac{2n}{n+2}}}. \nonumber
\end{eqnarray*}

According to the previous analysis, let
$$
f(s)=\dfrac{1}{m-1}M_0^{\frac{2n-m(n+2)}{n-2}}s-\dfrac{C(n)}{2(n-2)n\alpha(n)}s^{\frac{n-2}{n(m-1)}}.
$$
We now have a lower bound of the first part of free energy, i.e. $f\Big(\|\rho\|^{\frac{2n(m-1)}{n-2}}_{L^{\frac{2n}{n+2}}}\Big)\leq{\mathcal F}_1(\rho)$.

\begin{lem} \label{lowbounded}
If the initial free energy ${\mathcal F}(\rho_0)<{\mathcal F}^*:=f(s^*)$ and $\|\rho_0\|_{L^{\frac{2n}{n+2}}}>(s^*)^{\frac{n-2}{2n(m-1)}}$, let $\rho(x,t)$ be a solution of problem \eqref{equs},
 then there exists a constant $\mu_2>1$ such that
$$
\|\rho(\cdot, t)\|_{\frac{2n}{n+2}}> (\mu_2s^*)^{\frac{n-2}{2n(m-1)}}, \mbox{ for all } t>0,
$$
where $s^*$ is the maximum point of $f(s)$:
\begin{eqnarray*}
s^*=\Big(\dfrac{2n^2\alpha(n)M_0^{\frac{2n-m(n+2)}{n-2}}}{C(n)}\Big)^{\frac{n(m-1)}{2n-2-mn}}.
\end{eqnarray*}
\end{lem}

\begin{proof}
Notice that $1<m<2-\frac{2}{n}$ implies $\frac{n-2}{n(m-1)}>1$, we know that $f(s)$ is a strictly concave function in $0<s<\infty$. Directly calculation shows that
$$
f'(s)=\dfrac{1}{m-1}M_0^{\frac{2n-m(n+2)}{n-2}}-\dfrac{C(n)}{2(n-2)n\alpha(n)}\frac{n-2}{n(m-1)}s^{\frac{2n-2-mn}{n(m-1)}}.
$$
As a consequence, $s^*$ is a unique maximum point of $f(s)$.
Therefore the important property of $f$ is that $f(s)$ is monotone increasing for $0<s<s^*$, while $f(s)$ is monotone decreasing for $s>s^*$.

In the case that initial free energy ${\mathcal F}(\rho_0)<f(s^*)$, we can make it even smaller, i.e. there is a $\delta<1$ such that ${\mathcal F}(\rho_0)<\delta f(s^*)$.

Combining all the facts we know, including the interpolation, Haddy-Littlewood-Sobolev inequality and the monotonicity of free energy, we have
\begin{eqnarray*}
f\Big(\|\rho\|^{\frac{2n(m-1)}{n-2}}_{L^{\frac{2n}{n+2}}}\Big)\leq{\mathcal F}_1(\rho)\leq {\mathcal F}(\rho)\leq {\mathcal F}(\rho_0)<\delta f(s^*).
\end{eqnarray*}
If initially $\|\rho_0\|^{\frac{2n(m-1)}{n-2}}_{L^{\frac{2n}{n+2}}}>s^*$, due to the fact that $f(s)$ is increasing in $s>s^*$, there exists a $\mu_2>1$ such that  $\|\rho\|^{\frac{2n(m-1)}{n-2}}_{L^{\frac{2n}{n+2}}}>\mu_2 s^*$.
\end{proof}

\subsection{Time derivative of second moment}$\newline$
In this subsection, we will focus on studying the time evolution of the second moment. The following lemma is obtained from Lemma \ref{lowbounded}.
\begin{lem}\label{second moment}
If ${\mathcal F}(\rho_0)<{\mathcal F}^*:=f(s^*)$ and $\|\rho_0\|_{\frac{2n}{n+2}}> (s^*)^{\frac{n-2}{2n(m-1)}}$, $\rho$ is a solution of \eqref{equs}, then
\begin{eqnarray}
\dfrac{dm_2(t)}{dt} <0.
\end{eqnarray}
\end{lem}

\begin{proof}
By direct calculation, we have
\begin{eqnarray*}
\dfrac{dm_2(t)}{dt}=\Big(2n-\dfrac{2(n-2)}{m-1}\Big) \dint_{\mathbb{R}^n}\rho^m dx +2(n-2)\mathcal{F}(\rho).
\end{eqnarray*}
The restriction on $m<2-\frac{2}{n}$ gives that $2n-\frac{2(n-2)}{m-1}<0$. Then by using interpolation inequality, the decreasing properties of free energy and Lemma \ref{lowbounded} with $\mu_2>1$, we have
\begin{eqnarray*}
\dfrac{dm_2(t)}{dt}&\leq & \Big(2n-\dfrac{2(n-2)}{m-1}\Big) M_0^{\frac{(\theta-1)m}{\theta}} \|\rho\|_{L^{\frac{2n}{n+2}}}^{\frac{m}{\theta}} +2(n-2)\mathcal{F}(\rho_0)\\
&<& \Big(2n-\dfrac{2(n-2)}{m-1}\Big) M_0^{\frac{(\theta-1)m}{\theta}} \mu_2 s^* +2(n-2)f(s^*)\\
&=& \Big(2n-\dfrac{2(n-2)}{m-1}\Big) M_0^{\frac{(\theta-1)m}{\theta}} (\mu_2-1) s^*+\Big(2n-\dfrac{2(n-2)}{m-1}\Big) M_0^{\frac{(\theta-1)m}{\theta}} s^* \\
&&\hspace{2cm}+2(n-2)\Big(\dfrac{1}{m-1}M_0^{\frac{(\theta-1)m}{\theta}}s^* -\dfrac{C(n)}{2(n-2)n\alpha(n)}(s^*)^{\frac{2\theta}{m}}\Big)\\
&=& \Big(2n-\dfrac{2(n-2)}{m-1}\Big) M_0^{\frac{(\theta-1)m}{\theta}} (\mu_2-1) s^*+2n M_0^{\frac{(\theta-1)m}{\theta}}s^* -\dfrac{C(n)}{n\alpha(n)}(s^*)^{\frac{2\theta}{m}}\\
&=&\Big(2n-\dfrac{2(n-2)}{m-1}\Big) M_0^{\frac{(\theta-1)m}{\theta}} (\mu_2-1) s^*<0.
\end{eqnarray*}
where the last second equation follows from the definition of $s^*$.
\end{proof}

\subsection{The proof on the blow-up result in Theorem \ref{thm}}$\newline$
From Lemma \ref{second moment}, we know that there exists a finite time $T$ such that
$$
\lim_{t\rightarrow T} m_2(t)=0.
$$

The relation between the second moment and $L^m$ norm of $\rho$ can be obtained by using H\"{o}lder's inequality, $\forall R>0$, we have
 \begin{eqnarray*}
&&\int_{ {\mathbb R}^n}\rho( x)d x\leq\int_{B_R}\rho( x)d x+\int_{B^c_R}\rho( x)d x
\leq CR^{n(m-1)/m}\|\rho\|_{L^m}+\frac{1}{R^2}m_2(t).
 \end{eqnarray*}
Now by choosing $R=(\frac{m_2(t)}{C\|\rho\|_{L^m}})^{\frac{m}{(m-1)n+2m}}$, we have
$$\|\rho\|_{L^1}\leq C\|\rho\|^{\frac{2m}{(m-1)n+2m}}_{L^m}m_2(t)^{\frac{n(m-1)}{(m-1)n+2m}}.$$
Consequently, there exists $T^*\leq T$ such that $\lim_{t\rightarrow T^*}\|\rho\|_{L^m}=\infty$.

\section*{Appendix}

In Example \ref{example1}, we gave an initial data of the system with small mass and showed that the solution must blow up in finite time according to the main result of this paper. Here in this appendix, we will give a detailed calculation for the quantities appeared in Example \ref{example1} to make sure that the assumptions in theorem \ref{thm} satisfied.

For given $\varepsilon_0>0$ small, let the initial data be
 \begin{eqnarray}\label{rho0}
\rho_0(x)=\left\{\begin{array}{ll} \varepsilon_0 \frac{K^n}{\alpha(n)}, & |x|\leq \frac{1}{K},\\ 0, & |x|>\frac{1}{K},\end{array}\right.
 \end{eqnarray}
where $\alpha(n)$ is the volume of $n$ dimensional unit ball, and $K$ will be determined later.

First of all, since $\|\rho_0\|_{L^{\frac{2n}{n+2}}}=\varepsilon_0 \Big(\frac{K^n}{\alpha(n)}\Big)^{\frac{n-2}{2n}}$, to prove \eqref{examrelmass}, i.e. $\|\rho_0\|_{L^{\frac{2n}{n+2}}}>(s^*)^{\frac{n-2}{2n(m-1)}}$, it is necessary to show
\begin{eqnarray}\label{relation2}
 \varepsilon_0^{1+\frac{m(n+2)-2n}{2(2n-2-mn)}} K^{\frac{n-2}{2}}>(\alpha(n))^{\frac{n-2}{2n}}\Big(\frac{2n^2\alpha(n)}{C(n)}\Big)^{\frac{n-2}{2(2n-2-mn)}}.
\end{eqnarray}
Notice that $n>2$, there exists a constant $K_1>0$ such that for all $K>K_1$, the formula (\ref{relation2}) is true.

The corresponding initial free energy is
\begin{eqnarray*}
&&{\mathcal F}(\rho_0)=\frac{1}{m-1}\int_{\mathbb{R}^n}\rho_0^m dx-\frac{1}{2(n-2)n\alpha(n)}\int\int_{\mathbb{R}^n\times\mathbb{R}^n}\frac{\rho_0(x)\rho_0(y)}{|x-y|^{n-2}}dxdy\\
&=&\frac{1}{m-1}\int_{|x|\leq \frac{1}{K}}\varepsilon_0^m\big(\frac{K^n}{\alpha(n)}\big)^mdx-\frac{1}{2(n-2)n\alpha(n)}\int_{|x|\leq \frac{1}{K}}\int_{|y|\leq \frac{1}{K}}\frac{\big(\varepsilon_0\frac{K^n}{\alpha(n)}\big)^2}{|x-y|^{n-2}}dxdy
\\
&\leq& \frac{\varepsilon_0^m }{m-1}K^{n(m-1)}(\alpha(n))^{1-m}-\frac{1}{2(n-2)n\alpha(n)}\int_{|x|\leq \frac{1}{K}}\int_{|y|\leq \frac{1}{K}}\frac{\big(\varepsilon_0\frac{K^n}{\alpha(n)}\big)^2}{(|x|+|y|)^{n-2}}dxdy\\
&\leq& \frac{\varepsilon_0^m }{m-1}K^{n(m-1)}(\alpha(n))^{1-m}-\frac{1}{2(n-2)n\alpha(n)}\int_{|x|\leq \frac{1}{K}}\int_{|y|\leq \frac{1}{K}}\frac{\big(\varepsilon_0\frac{K^n}{\alpha(n)}\big)^2}{(\frac{2}{K})^{n-2}}dxdy\\
&=& \frac{\varepsilon_0^m }{m-1}K^{n(m-1)}(\alpha(n))^{1-m}-\frac{2^{2-n}}{2(n-2)n\alpha(n)}\varepsilon_0^2 K^{n-2}.
\end{eqnarray*}
To show that \eqref{examrelentropy} is true, it is necessary to show that
 \begin{eqnarray}\label{relation1}
 &&\varepsilon_0^{m+\frac{m(n+2)-2n}{2n-2-mn}} K^{n(m-1)}(\alpha(n))^{1-m}\nonumber\\
 &<&\frac{(m-1)2^{2-n}}{2(n-2)n\alpha(n)}\varepsilon_0^{2+\frac{m(n+2)-2n}{2n-2-mn}} K^{n-2}+\frac{2-\frac{2}{n}-m}{1-\frac{2}{n}}\Big(\frac{2n^2\alpha(n)}{C(n)}\Big)^{\frac{n(m-1)}{2n-2-mn}}.
\end{eqnarray}
Notice that $m<2-\frac{2}{n}$ implies $n(m-1)<n-2$. Thus there exists a constant $K_2>0$ such that when $K>K_2$, (\ref{relation1}) holds.

Hence taking $K_0=\max\{K_1,K_2\}$, we know that when $K>K_0$, the initial data satisfies blow-up condition in Theorem \ref{thm}.

\end{document}